\documentclass[runningheads]{llncs}

\usepackage{graphicx}
\usepackage{url}
\usepackage{balance}

\usepackage{mathtools}
\usepackage{bm}
\usepackage[mathscr]{eucal}

\usepackage{booktabs}
\usepackage{amsmath, amsfonts} 
\usepackage[linesnumbered,ruled,algonl,vlined,noend]{algorithm2e}
\usepackage{epsfig}
\usepackage{hyperref}
\usepackage{url}
\usepackage{enumitem}
\usepackage{tabularx}
\usepackage{multirow}
\usepackage{verbatim}
\usepackage{float}
\usepackage{tikz}
\usepackage{amssymb}
\usepackage{subfigure}

\DeclareMathOperator*{\argmax}{argmax}

\newcommand{\G}{\mathcal{G}} 

\newcommand{\Ps}{\mathcal{P}} 

\newcommand{\ps}{p} 

\newcommand{\Route}{\mathcal{R}}

\newcommand{\Number}{\mathcal{N}}

\newcommand{\Cluster}{\mathcal{C}}

\newcommand{\num}{n}

\newcommand{\rt}{r} 

\newcommand{\bs}{b} 

\newcommand{\Bs}{\mathcal{B}} 

\newcommand{\dt}{dt} 

\newcommand{\OPT}{\mathcal{O}} 

\newcommand{\Se}{\mathcal{S}} 

\newcommand{\Fr}{\mathcal{F}} 

\newcommand{\fr}{f}

\newcommand{\Problem}{FAST\xspace}

\newcommand{\smallVspace}{\vspace{-0.1in}}
\newcommand{\bigVspace}{\vspace{-0.2in}}

\newcommand{\fixint}{FixInterval\xspace}
\newcommand{\topk}{Top-$k$\xspace}
\newcommand{\greedy}{Greedy\xspace}
\newcommand{\partitiongreedy}{PartGreedy\xspace}
\newcommand{\partition}{BusRoutePartitioning\xspace}
\newcommand{\progressivegreedy}{ProGreedy\xspace}
\newcommand{\progressivepartgreedy}{ProPartGreedy\xspace}

\begin{document}
\title{Bus Frequency Optimization: When Waiting Time Matters in User Satisfaction}
\titlerunning{Bus Frequency Optimization}
\author{Songsong Mo\inst{1}\and
Zhifeng Bao\inst{2} \and
Baihua Zheng\inst{3} \and
Zhiyong Peng \inst{1}\thanks{Zhiyong Peng is the corresponding author.}}
%
%
\institute{School of Computer Science,Wuhan University, Hubei, China\\ \email{\{songsong945, peng\}@whu.edu.cn} \and
RMIT University, Melbourne, Australia\\ \email{zhifeng.bao@rmit.edu.au} \and
Singapore Management University, Singapore, Singapore\\ \email{bhzheng@smu.edu.sg}}
%
\maketitle              
\smallVspace
\begin{abstract}

	Reorganizing bus frequency to cater for the actual travel demand can save the cost of the public transport system significantly. Many, if not all, existing studies formulate this as a bus frequency optimization problem which tries to minimize passengers' average waiting time. However, many investigations have confirmed that the user satisfaction drops faster as the waiting time increases. Consequently, this paper studies the bus frequency optimization problem considering the user satisfaction. Specifically, for the first time to our best knowledge, we study how to schedule the buses such that the total number of passengers who could receive their bus services within the waiting time threshold is maximized. We prove that this problem is NP-hard, and present an index-based algorithm with $(1-1/e)$ approximation ratio. By exploiting the locality property of routes in a bus network, we propose a partition-based greedy method which achieves a $(1-\rho)(1-1/e)$ approximation ratio. Then we propose a progressive partition-based greedy method to further improve the efficiency while achieving a $(1-\rho)(1-1/e-\varepsilon)$ approximation ratio. Experiments on a real city-wide bus dataset in Singapore verify the efficiency, effectiveness, and scalability of our methods.
	\smallVspace
	\smallVspace
\end{abstract}

\vspace{-0.8cm}
\section{Introduction}
\smallVspace

Public transport and the services delivered by buses are essential to our daily life. Bus services provide us with the capability to move around, which shapes where we can work and live, where we shop and how we spend our leisure time. In this paper, we focus on bus frequency design which plays a very important role in urban public transport systems, as reorganizing bus frequencies to meet the actual travel demands is expected to achieve significant savings in cost. Taking New York City as an example, the cost of each bus is around \$550,000 and the operating cost of transit agencies reaches \$215 per hour\footnote{\url{https://www.liveabout.com/bus-cost-to-purchase-and-operate-2798845}}. If we re-organize the bus frequencies based on real travel demands and save 10\% bus departures, we can save \$20 operating costs per hour and \$55,000 per vehicle. 


In the literature, there are many studies focusing on the problem of bus frequency optimization. Most of them share a common objective, which is to 
minimize the \emph{average} travel cost (in terms of waiting time) of passengers~\cite{scheele1980supply,constantin1995optimizing,gao2004continuous,DBLP:journals/eor/MartinezMU14,8681337}. Moreover, their solutions are usually heuristic rather than approximate (with theoretical guarantees). However, most, if not all, existing works ignore an important aspect, the \emph{user satisfaction}. Many studies have confirmed that the user satisfaction drops faster as the waiting time increases~\cite{antonides2002consumer,kong2007correlates}. Motivated by this finding, we aim to schedule the buses in a way to serve more passengers within a given waiting time threshold $\theta$ but not to minimize the average waiting time. In addition, our algorithms are adaptive to cater for different settings of $\theta$.


We call this novel problem as  Satis\underline{FA}ction-Boo\underline{ST} Bus Scheduling (\Problem). Given a bus database $\Bs$, a bus route database $\Route$, a passenger database $\Ps$, and a vector $\Number$ $\langle n_1, n_2, \cdots, n_i, \cdots, n_{|\Route|} \rangle$ that specifies the expected number of bus departures for each bus route, it chooses $n_i$ buses for each route $\rt_i \in \Route$ such that the whole bus system is able to satisfy the most passengers. The analysis shows that the objective function of \Problem is submodular and \Problem is NP-hard. 


To resolve the \Problem problem, we develop a range of approximate algorithms with non-trivial theoretical guarantees. First, we propose an index-based greedy method (\greedy), which can provide $(1-1/e)$ approximation factor for \Problem as the baseline, and two enhanced versions, namely \partitiongreedy and \progressivepartgreedy. \partitiongreedy is inspired from \cite{DBLP:conf/kdd/ZhangBLLZP18} and by the fact that a bus network is designed to cover different parts of the city and it tries to avoid unnecessary overlapping among routes~\cite{fletterman2009designing,DBLP:journals/tkde/WangBCSC18}. It adopts a partitioning algorithm to divide the bus network into several disjoint partitions. Accordingly, it invokes local greedy search within each partition, which effectively reduces the computation cost of the original greedy algorithm. On the other hand,  \progressivepartgreedy adopts a different strategy to address the efficiency issue. Instead of finding one bus that contributes the most to the objective function in each iteration of the local greedy search, it fetches multiple buses in each iteration of the local greedy search to cut down the total number of iterations required. Meanwhile, \progressivepartgreedy has a tunable parameter that could determine roughly how many buses could be fetched in each iteration and hence provide a trade-off between efficiency and effectiveness.
%

In summary, we make the following contributions.
\begin{itemize}\smallVspace
	\item We propose and study the \Problem problem. To the best of our knowledge, this is the first study on bus frequency optimization that considers user satisfaction. We prove  that the objective function of \Problem is  monotone and submodular, and \Problem is NP-hard.
	\item We propose an index-based greedy method (\greedy), a partition-based greedy method (\partitiongreedy) and a progressive partition-based greedy method (\progressivepartgreedy) to solve the \Problem problem efficiently. They can achieve an approximation ratio of $(1 - \frac{1}{e})$, $(1-\rho)(1 -  \frac{1}{e})$, and $(1-\rho)(1- \frac{1}{e}-\varepsilon)$ respectively, where $\rho$ and $\varepsilon$ are the user-defined parameters. 
	\item We conduct extensive experiments on real-world bus route and bus touch-on/touch-off records in Singapore (396 routes, 28 million trip records of one week) to demonstrate the effectiveness, efficiency and scalability of our methods.
	\smallVspace
\end{itemize}


\smallVspace
\section{Related Work} \label{related_work}
\smallVspace

In this section, we will review existing related work and report the difference between this work and existing ones.

We divide the literature into two categories based on the overall optimization objective. One is called the travel time driven bus frequency optimization problem (Travel-BFO), which aims to minimize the average/total travel time of passengers for either one bus route or a bus route network, based on passenger demands. It treats each ride as a new trip. Another is called the transfer time driven bus frequency optimization problem (Transfer-BFO), which aims to minimize the total transfer time of the transfer passengers.

\noindent
\textbf{Travel-BFO.} 
Here, the passenger demands are usually abstracted as origin-destination (OD) pairs.  
The model proposed in \cite{scheele1980supply} treats the travel time of passengers as an aggregation of the walking time, the waiting time, and the on-board travel time. 
The problem is usually formulated as a nonconvex objective function with linear or convex constraints. 
In~\cite{constantin1995optimizing}, it is modeled as a nonlinear bilevel problem: the upper level represents the planner who wants to ensure minimal total travel time under fleet size constraints; the lower level represents the users who act by minimizing the travel time. 
In~\cite{gao2004continuous}, a multi-objective model is proposed, seeking to minimize the overall travel time of the users and the operational cost of the operators (assumed to be linearly proportional to the frequencies). 
%
%
Mart{\'{\i}}nez et al.~\cite{DBLP:journals/eor/MartinezMU14} study the transit frequency optimization problem to determine the time interval between subsequent buses for a set of bus lines. They propose a mixed-integer linear programming (MILP) formulation for an existing bilevel model~\cite{constantin1995optimizing}, and present a metaheuristic method.
%
A new model considering user behavior is proposed in~\cite{8681337}. It assigns a user's trip to three stages (pre-trip, on-board and end-trip) 
and aims to minimize users’ total travel costs of the objective bus line. 
%

\noindent\underline{\textit{Differences}}. Although different bus frequency optimization models have been proposed, they share a very similar optimization objective, i.e., minimizing the average/total travel cost of passengers. Different from the above literature, we aim to improve the \emph{overall passenger satisfaction} by scheduling the buses such that they can serve more passengers within the given waiting time threshold. 
Our work is mainly motivated by the following two findings. \emph{First}, waiting time has a direct impact on the user satisfaction, as evident by many studies~\cite{kong2007correlates,antonides2002consumer}. \emph{Second}, the waiting time threshold is tunable, hence the bus company can adjust thresholds to cater to various concerns on budget, government needs, passengers' tolerance of waiting, etc. 
%

\vspace{.2em}
\noindent
\textbf{Transfer-BFO.} Transfer time driven bus frequency optimization problem is an extension of single bus route timetabling. It determines the departure time of each trip of all lines in the bus network with the consideration of passenger transfer activities at transfer stations~\cite{ibarra2015planning}. 

This problem is modeled by mixed integer programming models to maximize the number of synchronized bus arrivals at transfer nodes~\cite{ceder2001creating}. 
Ibarra-Rojas~et~al.~\cite{ibarra2012synchronization} extend
the work of Ceder et al.~\cite{ceder2001creating} to address a flexible Transfer-BFO problem with almost evenly spaced departures and preventing bus bunching. The model proposed in \cite{shafahi2010practical} tries to minimize the total transfer time experienced by passengers. Parbo et al.~\cite{parbo2014user} studied a bi-level bus timetabling problem to minimize the weighted transfer waiting time of passengers, and a Tabu Search algorithm was applied to solve the bilevel model. Recently a nonlinear mixed integer-programming model is proposed to maximize the number of total transferring passengers with small excess transfer time~\cite{wu2019combining}.

\noindent\underline{\textit{Differences}}. The above studies on the Transfer-BFO problem mainly focus on minimizing the total transfer cost for passengers on transfer, which can only improve the satisfaction of the transfer passengers. In contrast, our problem aims to improve overall passenger satisfaction by serving them within a given time threshold.

For all the above work in both categories, despite the difference, all existing approaches only propose heuristic methods without theoretical guarantees, while we propose algorithms with non-trivial theoretical guarantees.

\vspace{-0.2cm}
\section{Problem Formulation} \label{sec_3}
\smallVspace

In a bus route database $\Route$, a route $\rt$ is a sequence of bus stations ($s_1$, $s_2$, $\cdots$, $s_i$, $\cdots$, $s_m$), where $s_i$ is a bus station represented by (latitude, longitude). In a passenger database $\Ps$, a passenger $\ps \in \Ps$ is in form of a tuple $\{s_b,s_e,t\}$, where $s_b$ denotes the boarding station, $s_e$ denotes the alighting station, and $t$ denotes the time when $\ps$ reaches $s_b$. A bus $\bs_{ij}$ is in form of a tuple $\{\rt_i,\dt_j\}$, where $\rt_i$ and $\dt_j$ denote the bus service route and the departure time from $\rt_i.s_1$ respectively.

\begin{definition}\label{defn:serve}
	We define that a bus $\bs_{ij}$ can serve a passenger $\ps$, if $\rt_i$ contains $\ps.s_b$ and $\ps.s_e$ in order, and $0\le \dt_j+T(\rt_i.s_1,\ps.s_b)-t \le \theta$, where $T(\rt_i.s_1,\ps.s_b)$ denotes the travel time required by bus $\bs_{ij}$ from $\rt_i.s_1$ to $\ps.s_b$  via the bus route $\rt_i$, and $\theta$ is a given waiting time threshold.
\end{definition}

There are multiple ways available to approximate $T(\rt_i.s_1,\ps.s_b)$. In this paper, we utilize the historical average travel time from $\rt_i.s_1$ to $\ps.s_b$ via the route $\rt_i$ to compute $T(s_1,s_b)$.
Based on Definition~\ref{defn:serve}, we formally introduce $\Se(\bs_{ij},\ps_k)$ to  denote the service of $\bs_{ij}$ to $\ps_k$, as presented in Equation~(\ref{equ:Se}). 
\begin{equation}\label{equ:Se}
\Se(\bs_{ij},\ps_k) = \left\{ \begin{array}{l}
1  	{\textrm{ if $\bs_{ij}$ can serve $p_k$}}\\
0	{\textrm{ otherwise}}
\end{array} \right. 
\end{equation}

Next, we introduce the concept of bus service frequency in Definition~\ref{defn:frequency}. Let the bus service frequency $\Fr$ for $\Route$ be a set, with each element $\fr_i \in \Fr$ corresponding to a bus route $\rt_i \in \Route$, i.e., $\Fr=\{\cup_{\forall \rt_i \in \Route}\fr_i$\}. 
Then, the service of $\Fr$ to a passenger $\ps_k$ can be computed by Equation~(\ref{equ:Se_F}). Note $\Se(\Fr,\ps_k) = 1$ as long as any $\bs_{ij} \in \Fr$ can serve $p_k$; otherwise, $\Se(\Fr,\ps_k) = 0$.
\begin{equation}\label{equ:Se_F}
\Se(\Fr,\ps_k) = 1 - \prod\nolimits_{{\bs_{ij}} \in \Fr} {(1 - \Se(\bs_{ij},\ps_k))} 
\end{equation}

\begin{definition}\label{defn:frequency}
	A bus service frequency ($\fr_i$) for $\rt_i$ refers to a set of buses ($\bs_{i1}$, $\bs_{i2}$, $\cdots$, $\bs_{in_i}$) that serve the route $\rt_i$, where $n_i$ $(n_i\ge1)$ denotes the total number of bus departures corresponding to the route $\rt_i$ within a day.
\end{definition}

Next, we formulate our problem in Definition~\ref{defn:problem} and show its NP-hardness. Note that we ignore the passenger capacity of the bus in our problem definition. 
\smallVspace
\begin{definition}[Satis\underline{FA}ction-Boo\underline{ST} Bus Scheduling (\Problem)]\label{defn:problem}
	Given a bus route database $\Route$, a passenger database $\Ps$, a waiting time threshold $\theta$, and a vector $\Number \langle n_1$, $n_2$, $\cdots$,$n_i$, $\cdots$, $n_{|\Route|}\rangle$ where $n_i$ $(\geq 1)$ denotes the total number of bus departures of bus route $\rt_i\in \Route$, we output a bus service frequency $\Fr$ which can maximize $\G(\Fr) = \sum\nolimits_{{\ps_k} \in \Ps} \Se(\Fr,\ps_k)$, 
	where $\G(\Fr)$ denotes the total number of passengers served by $\Fr$.
\end{definition}

\begin{theorem}\label{theorem:G}
	The objective function $\G$ of \Problem is monotone and submodular.
\end{theorem}

\begin{proof}
	We skip the proof of the monotonicity of $\G$ as it is straightforward. In the following, we prove that $\G$ is submodular.
	Let $V \subseteq T\subset \mathcal{B}$, where $\mathcal{B}$ denotes the universe of buses, and $b$ refers to a bus in $\mathcal{B}\backslash T$. According to \cite{DBLP:journals/mp/NemhauserWF78}, $\G(V)$ is submodular if it satisfies: $\G(V  \cup b)-\G(V) \geq \G(T \cup b)-\G(T)$. To facilitate the proof, we define $V_b=V \cup b$ and $\G_b(V)=\G(V \cup b)-\G(V)$. Then, we have:
	\begin{equation}\label{sub:0}
	\begin{aligned}
	\G_b(V)- \G_b(T)&= (\sum\nolimits_{{\ps_k} \in \Ps} \Se(V_b,\ps_k) - \sum\nolimits_{{\ps_k} \in \Ps} \Se(V,\ps_k)) \\
	- &(\sum\nolimits_{{\ps_k} \in \Ps} \Se(T_b,\ps_k) - \sum\nolimits_{{\ps_k} \in \Ps} \Se(T,\ps_k))\\
	&=\sum\nolimits_{{\ps_k} \in \Ps} (\Se(V_b,\ps_k) - \Se(V,\ps_k)-\Se(T_b,\ps_k)+\Se(T,\ps_k)).\\
	\end{aligned}
	\end{equation}
	To show the submodularity of $\G$, we first prove Inequality~(\ref{sub:1}).
	\begin{equation}\label{sub:1}
	\Se(V_b,\ps_k) - \Se(V,\ps_k)-\Se(T_b,\ps_k)+\Se(T,\ps_k)\geq 0
	\end{equation}
	According to whether $\ps_k$ can be served by buses in $V$ or buses in $T\backslash V$ or bus $b$, there are in total four cases corresponding to Inequality~(\ref{sub:1}).
	\underline{Case 1}: $\ps_k$ can be served by a bus $b_0 \in V$. Then we have $\Se(V,\ps_k)=\Se(V_b,\ps_k)=\Se(T,\ps_k)=\Se(T_b,\ps_k)=1$, because $V \subset V_b$ and $V \subseteq T \subset T_b$. Thus, $\Se(V_b,\ps_k) - \Se(V,\ps_k)-\Se(T_b,\ps_k)+\Se(T,\ps_k)=0$. 
	\underline{Case 2}: $\ps_k$ cannot be served by any bus $b_0 \in V$ but it can be served by a bus $b_1 \in T\backslash V$. Then we have $\Se(V,\ps_k) =0$, $\Se(V_b,\ps_k)\geq0$ and $\Se(T,\ps_k)=\Se(T_b,\ps_k)=1$. Thus, $\Se(V_b,\ps_k) - \Se(V,\ps_k)-\Se(T_b,\ps_k)+\Se(T,\ps_k)\geq0$. \underline{Case 3}: $\ps_k$ cannot be served by any bus $b_0 \in T$ and can be served by the bus $b$. Then we have $\Se(V,\ps_k)=\Se(T,\ps_k)=0$ and $\Se(V_b,\ps_k)=\Se(T_b,\ps_k)=1$. Thus, $\Se(V_b,\ps_k) - \Se(V,\ps_k)-\Se(T_b,\ps_k)+\Se(T,\ps_k)=0$. 
	\underline{Case 4}: $\ps_k$ cannot be served by any bus $b_0 \in T$ or the bus $b$. Then we have $\Se(V,\ps_k)$=$\Se(V_b,\ps_k)$=$\Se(T,\ps_k)=\Se(T_b,\ps_k)=0$. Thus, $\Se(V_b,\ps_k) - \Se(V,\ps_k)-\Se(T_b,\ps_k)+\Se(T,\ps_k)=0$. The above shows the correctness of Inequality~(\ref{sub:1}).
	Based on Equation~(\ref{sub:0}) and Inequality~(\ref{sub:1}), we have $\G_b(V)- \G_b(T)\geq 0$ and hence $\G$ is a submodular function. $\blacksquare$
\end{proof}

\smallVspace
\begin{theorem} \label{np}
	The \Problem problem is NP-hard.
\end{theorem}
\smallVspace
\begin{proof}
	It is worth noting that the minimum unit of time is second in daily life. Therefore, $\Bs$ is a finite set. Based on this, we prove it by reducing the Set Cover problem to the \Problem problem. In the Set Cover problem, given a collection of subsets $S_1$, $\cdots$, $S_i$, $\cdots$, $S_j$ of a universe of elements $U$, we wish to know whether there exist $k$ of the subsets whose union is equal to $U$. We map each element in $U$ in the Set Cover problem to each passenger in $\Ps$, and map each subset $S_i$ to the set of passengers server by a bus $b \in \Bs$. Consequently, if all passengers in $U$ are served by $S$, the total number of passengers served by $S$ is $|U|$. Subsequently, $\num = \sum_{i = 1}^{|\Route|} {{n_i}}$ is set to $k$ (selecting $k$ buses). The Set Cover problem is equivalent to deciding if there is a $k$-bus set with the maximum served passenger number $U$ in \Problem. As the Set Cover problem is NP-complete, the decision problem of \Problem is NP-complete, and the optimization problem is NP-hard. $\blacksquare$
\end{proof}

\setlength{\algomargin}{1.2em} 
\begin{algorithm}[t]
	\caption{\greedy$(\Bs, \Route, \Ps, \Number)$}
	\label{greedy}
	\begin{small}
		{\bf Input:} a bus database $\Bs$, a bus route database $\Route$, a passenger database $\Ps$, and a vector $\Number$ $\langle n_1, n_2, \cdots, n_{|\Route|} \rangle$
		
		{\bf Output:} a bus service frequency $\Fr$
		
		Initialize $\Fr \gets \phi$, $\num\gets \sum_{i = 1}^{|\Number|} {{n_i}}$
		
		Initialize a $|\Number|$-dimension vector $\langle k_1, k_2, \cdots, k_{|\Number|}\rangle$ with zero
		
		
		\For{$i \gets 1$ to $\num$}
		{	
			Select a bus $\bs_{jl} \gets  \mathop {\arg \max }_{\bs \in \Bs \backslash\Fr}(\G(\Fr \cup \bs)- \G(\Fr))$ \label{marginal}
			%
			
			$k_j++$	
			
			\If{$k_j \le n_j$}
			{\label{lineCheck:start}
				$\Fr \gets \Fr \cup \bs_{jl} $	
				
			}
			\If{$k_j \ge n_j$}
			{
				remove all the buses serving the route $j$ from $\Bs$
			}
			\label{lineCheck:end}
			%
			%
			%
		}
		\Return {$\Fr$}
	\end{small}
\end{algorithm}
\bigVspace

\smallVspace
\section{Basic Greedy Method} \label{sec_4}
\smallVspace

To address \Problem, we first present a baseline which extends the basic greedy method for the problem of submodular function maximization. To accelerate the marginal gain computation, we propose a mapping structure to index the bus and passenger database. The basic greedy method is guaranteed to achieve (1 - 1/$e$)-approximation, as proved by Nemhauser et al.~\cite{DBLP:journals/mp/NemhauserWF78}.

\smallVspace
\subsection{A Basic Greedy Method}
The pseudo-code of the greedy method is listed in Algorithm~\ref{greedy}. In each iteration, it selects a bus $\bs_{jl} \in \Bs \backslash\Fr$ with the largest marginal gain, such that $\bs_{jl}=\mathop {\arg \max }_{\bs \in \Bs \backslash\Fr}(\G(\Fr \cup \bs)- \G(\Fr))$, and inserts it to the current service frequency $\Fr$. 
In lines~\ref{lineCheck:start}-\ref{lineCheck:end}, it checks whether the number of bus departures of route $j$, which $\bs_{jl}$ serves, has reached the total number of bus departures required by this route. If so, it removes all buses serving the route $j$ from $\Bs$. Such an iteration is repeated $\num$ times, with $\num$ being the total number of bus departures required by all the bus routes. Finally, it returns $\Fr$ as the solution. 

\noindent
\textbf{Time Complexity.} In each iteration, Algorithm~\ref{greedy} needs to scan all the buses in $\Bs \backslash\Fr$ and computes their marginal gain to the chosen set. Each marginal gain computation
needs to traverse $\Ps$ once in the worst case. Thus, adding one bus into $\Fr$ takes $O(|\Ps|\cdot|\Bs|)$ time, and the total complexity is $O(\num\cdot |\Ps|\cdot|\Bs|)$.


\begin{figure}[t]
	\hspace{0.1in}
	\begin{minipage}[b]{0.15\textheight}
		\centering
		\begin{tabular}{|c||c|c|}
			\hline
			Bus List & $N_{ToBeServed}$ & $L_{P}$ \\ \hline
			$b_1$ & $3$ & $p_1,p_3, p_{|\Ps|}$ \\ \hline
			$b_2$ & $2$ & $p_1,p_2$ \\ \hline
			$b_3$ & $1$ & $p_3$ \\ \hline
			$\cdots$ & $\cdots$ & $\cdots$ \\ \hline
			$b_{|\Bs|}$ & $1$ & $p_2$ \\ \hline
		\end{tabular}
		\smallVspace
		\caption{Forward list}
		\label{fig:forward_list}
	\end{minipage}
	\hspace{0.8in}
	\begin{minipage}[b]{0.15\textheight}
		\centering
		\begin{tabular}{|c||c|c|}
			\hline
			Passenger List & $IsServed$ & Optional Buses \\ \hline
			$p_1$ & $false$ & $b_1,b_2$ \\ \hline
			$p_2$ & $false$ & $b_2,b_{|\Bs|}$ \\ \hline
			$p_3$ & $false$ & $b_1,b_3$ \\ \hline
			$\cdots$ & $\cdots$ & $\cdots$ \\ \hline
			$p_{|\Ps|}$ & $false$ & $b_1$ \\ \hline
		\end{tabular}
		\smallVspace
		\caption{Inverted list}
		\label{fig:inverted_list}
	\end{minipage}
	\bigVspace
\end{figure}

\smallVspace
\subsection{Index for Efficient Marginal Gain Computation}

To accelerate the marginal gain computation, which is the main bottleneck of Algorithm~\ref{greedy}, we propose two mapping indexes, \emph{forward list} and \emph{inverted list} as shown in Fig.~\ref{fig:forward_list} and Fig.~\ref{fig:inverted_list} respectively. The former is for buses $\bs_i \in \Bs$, maintaining a list of passengers $L_{P}$ that could be served by bus $\bs_i$. Note that a passenger could be served by multiple buses. To avoid counting the same passenger multiple times when we calculate the marginal gain, we maintain another parameter $N_{ToBeServed}$ to capture the number of passengers in $L_{P}$ that are still waiting for services. The initial value of $N_{ToBeServed}$ is set to be the cardinality of $L_{P}$, and its value will be reduced every time when a passenger in $L_{P}$ is served by another bus. The latter is for passengers $\ps \in \Ps$, maintaining a list of buses that could serve the passenger $\ps$. The boolean $IsServed$ is to indicate whether any of the optional buses has been scheduled with an initial value being $false$. For example, if bus $\bs_1$ is selected, it could serve three passengers based on $N_{ToBeServed}$'s value associated with $b_1$ in forward list. Meanwhile, $IsServed$'s value of passengers in $L_{P}$ of $\bs_1$ (i.e., $p_1,p_3, p_{|\Ps|}$) will be changed to $true$, all the buses that could serve $p_1$ or $p_3$ or $p_{|\Ps|}$ have to update $N_{ToBeServed}$'s value to reflect the fact that some of their potential passengers have already been served. 
%

\smallVspace
\section{Partition-based Greedy Method} \label{sec_5}
\smallVspace


In practice, a bus network is designed to cover different parts of a city to meet residents' various travel demands. By design, it tries to avoid unnecessary overlapping among routes~\cite{fletterman2009designing,DBLP:journals/tkde/WangBCSC18}. For example, Figure~\ref{fig:bus-routes} plots three popular bus routes in Singapore. A passenger whose travel demand could be served by route 67 will not consider route 161 or route 147 as these routes have \emph{zero} overlap. This observation suggests that it might be unnecessary to scan the entire bus network when calculating the marginal gains of certain buses. This motivates us to design a partition-based greedy method. In the following, we first introduce a novel concept namely \emph{service overlap ratio} to guide the partitioning process, and then present the algorithm. 
%
%

\begin{figure}[t]
	\centering
	\smallVspace
	\subfigure[Bus Route 67]
	{\label{fig:bus-67}\includegraphics[width=0.3\textwidth]{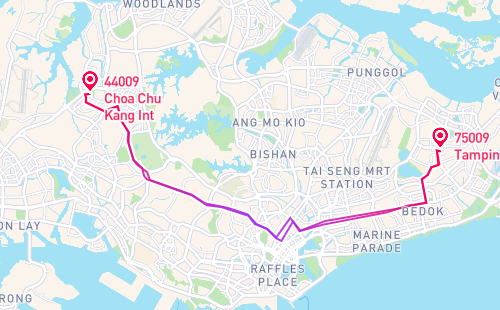}}
	\hspace{0.03\textwidth}
	\subfigure[Bus Route 147]
	{\label{fig:bus-147}\includegraphics[width=0.3\textwidth]{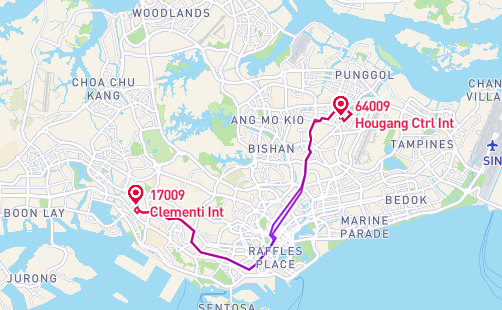}}
	\hspace{0.03\textwidth}
	\subfigure[Bus Route 161]
	{\label{fig:bus-147}\includegraphics[width=0.3\textwidth]{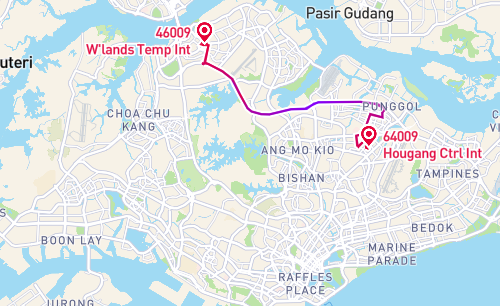}}
	\smallVspace
	\caption{Visualization of three popular bus routes in Singapore}
	\label{fig:bus-routes}
	\bigVspace
\end{figure}



Our main idea is to partition the bus routes (and buses) into disjoint clusters, and then use a divide-and-conquer strategy to find local optimal frequencies for routes in each partition. This approach is expected to reduce the time complexity of the basic greedy by a factor of $m^2$ with $m$ being the number of partitions. The speedup is contributed by the fact that it invokes the greedy algorithm for each cluster and hence it only needs to scan the buses and passengers corresponding to the routes in a cluster during the greedy search. Meanwhile, in term of accuracy, we introduce a novel concept called \emph{service overlap ratio} to achieve an approximation ratio with non-trivial theoretical guarantee, as shown later.

\begin{definition}[Partition]\label{defn:partition}
	A partition of  a set $S$ is denoted as a cluster set $\Cluster$=\{$\Cluster_1$, $\Cluster_2$, $\cdots$, $\Cluster_m$\}, where $m$ denotes the total number of clusters, such that $S = \cup_{i=1}^m \Cluster_i$, $\forall \Cluster_i \in \Cluster$, $\Cluster_i \ne \phi$, and $\forall \Cluster_i, \Cluster_j \in \Cluster$ with $i\ne j$, $\Cluster_i\cap \Cluster_j = \phi$.
\end{definition}

To better illustrate the service overlap ratio, we define a function \textsf{Serve}($P$,$R$) that takes a passenger set $P$ and a route set $R$ as inputs and returns the passengers in $P$ that could be served by any route in $R$ without considering the temporal factor. To be more specific, a passenger $\ps$ will be returned by \textsf{Serve}($P$,$R$) if there is a route $\rt_i \in R$ such that $\rt_i$ contains $\ps.s_b$ and $\ps.s_e$ in order, which is different from the ``bus serves passengers'' defined in Definition~\ref{defn:serve}. We name the set of passengers returned by \textsf{Serve}($P$,$R$) as the passenger pool w.r.t. bus routes $R$. 

As stated in Definition~\ref{ratio}, the service overlap ratio $\rho_i$ of a bus route cluster $\Cluster^R_i$ tries to measure the number of passengers in the passenger pool w.r.t. $\Cluster^R_i$ that actually also belong to the passenger pools w.r.t. other clusters.  Let $|A|$ denote the cardinality of the set $A$, and $\overline{\Fr}_i$ denote a bus service frequency returned by $\rm{\greedy}(\Cluster^B_i, \Cluster^R_i, \Cluster^P_i, N_{min})$. $\Cluster^B_i$, $\Cluster^R_i$, and $\Cluster^P_i$ refer to a cluster of buses, a cluster of routes and a cluster of passengers respectively, and $N_{min}$ refers to a $|\Cluster^R_i|$-dimensional vector in the form of $\langle n_{min}, n_{min}, \cdots, n_{min}\rangle$. The parameter $n_{min}$ is set to the minimum number of buses required by any route. Although there are different ways to quantify the overlaps between bus routes, we define $\rho_i$ in such a way that a partition-based greedy guided by $\rho_i$ can achieve a theoretical bound, as to be detailed next. 
%

\begin{definition}[Service overlap ratio]\label{ratio}
	Given a partition  $\Cluster^R$ of the original bus route database $\Route$, for a cluster $\Cluster^R_i$, the ratio $\rho_i$ of the service overlap between $\Cluster^R_i$ and the rest clusters is $ \frac{ \left| {\bigcup\nolimits_{{\Cluster^R_j} \in {\Cluster^R \backslash \Cluster^R_i}} \textsf{Serve}(\Ps, \Cluster^R_i) \cap \textsf{Serve}(\Ps, \Cluster^R_j)} \right|}{\G(\overline{\Fr}_i)}$. 
\end{definition}

\setlength{\algomargin}{1.2em} 
\begin{algorithm}[t]
	\caption{\partitiongreedy$(\Bs, \Route, \Ps, \Number, \rho)$ }
	\label{partitiongreedy}
	\begin{small}
		{\bf Input:} a bus database $\Bs$, a bus route database $\Route$, a passenger database $\Ps$, and a vector $\Number$ $\langle n_1, n_2, \cdots, n_{|\Route|} \rangle$, a controlling threshold $\rho$
		
		{\bf Output:} a bus service frequency $\Fr$
		
		initialize $\Cluster^R\gets \phi$, $\Cluster^B\gets \phi$, $S_P \gets \phi$, $n_{min}\gets Min_{1\le i \le |\Route|}n_i$, $\Fr \gets \phi$
		
		$(\Cluster^B, \Cluster^R) \gets \rm{\partition}(\Bs, \Route, n_{min}, \rho)$
		
		
		\For{each cluster $\Cluster^R_i \in \Cluster^R$}
		{	
			$S_P \gets \textsf{Serve}(\Ps, Cluster^R_i)$, $\Fr \gets \Fr \cup \rm{\greedy}(\Cluster^B_i, \Cluster^R_i, S_P, \Number)$\label{part:replace}
			
		}
		\Return {$\Fr$}
	\end{small}
\end{algorithm}

\noindent
\textbf{Partitioning of bus routes and buses}.  Algorithm~\ref{partition} lists the pseudo-code of a bus route partitioning method guided by service overlap ratio. It first partitions the routes using the finest granularity by forming a cluster for each bus route. Thereafter, it checks the service overlap ratio $\rho_i$ for each cluster $\Cluster^R_i$ and picks the one with the largest $\rho_i$, denoted as $\Cluster^R_k$, for expansion (Line \ref{code:find_k}). It selects the cluster $\Cluster^R_j$ that shares the largest common passenger pool with $\Cluster^R_k$(Line~\ref{code:find_j}) and merges $\Cluster^R_j$ with $\Cluster^R_k$ (Lines~\ref{code:merge} - \ref{code:merge_end}). Note that when cluster $\Cluster^R_k$ is expanded, let $\overline{\Fr}_k$ denote the new frequency returned by \rm{\greedy}$(\Cluster^B_k, \Cluster^R_k, \Ps, N_{min})$. $\G({\overline{\Fr}_k})$ is actually required when calculating $\rho_k$ for this expanded cluster, by Definition~\ref{ratio}. However, to reduce the computation cost and the complexity, we use $\mathcal{L}=max\{\G({\overline{\Fr}_k}) + \G({\overline{\Fr}_j})-|S_k\cap S_j|, \G({\overline{\Fr}_k}) , \G({\overline{\Fr}_j})\}$ as an approximation of $\G({\overline{\Fr}_k})$. According to our merger rules, $\mathcal{L}$ is a lower bound of $\G({\overline{\Fr}_k})$ and it does not affect the accuracy of our partition algorithm. This merge-and-expansion process continues until the $\rho_i$s associated with all the clusters $\Cluster^R_i$ fall below the input threshold $\rho$. 

\setlength{\algomargin}{1.2em} 
\begin{algorithm}[t]
	\caption{\partition$(\Bs, \Route, n_{min}, \rho$)}
	\label{partition}
	\begin{small}
		{\bf Input:} a bus database $\Bs$, a bus route database $\Route$, an integer $n_{min}$, and a controlling threshold $\rho$
		
		{\bf Output:} a partition $\Cluster^B$ of $\Bs$ and a partition $\Cluster^R$ of $\Route$
		
		\For{each bus route $\rt_i \in Route$}
		{
			initialize $\Cluster^R_i \gets \{\rt_i\}$, $\Cluster^B_i \gets \{b_{ab} \in \Bs |a = i\}$, $S_i \gets \textsf{Serve}(\Ps, Cluster^R_i)$
			
			$\overline{\Fr}_i \gets  \rm{\greedy}(\Cluster^B_i,\Cluster^R_i,\Ps, N_{min})$
		}
		
		initialize $\Cluster^R \gets \cup_{\rt_i \in \Route} \Cluster^R_i$
		
		
		\For{$\Cluster^R_i \in \Cluster^R$}{
			$\rho_i \gets { \left| {\bigcup\nolimits_{{\Cluster^R_j} \in {\Cluster^R \backslash \Cluster^R_i}} {{S_i \cap S_j}}} \right|}/{\G(\overline{\Fr}_i)}$
			
		}
		
		$k \gets 	\argmax_{\Cluster^R_k \in \Cluster^R} \rho_k$, $Max \gets \rho_{k}$ \label{code:find_k}
		
		\While{$Max> \rho$} 
		{\label{part:1}
			
			$j \gets \argmax_{\Cluster^R_j \in \Cluster^R\backslash \Cluster^R_k} |(S_j \cap S_k)|$ \label{code:find_j}
			
			$\Cluster^R_k\gets \Cluster^R_k \cup \Cluster^R_j$, $\Cluster^R \gets \Cluster^R -\Cluster^R_j$, 
			$\Cluster^B_k\gets \Cluster^B_k \cup \Cluster^B_j$, $\Cluster^B \gets \Cluster^B -\Cluster^B_j$ \label{code:merge}
			
			
			$\G({\overline{\Fr}_k}) \gets max\{\G({\overline{\Fr}_k}) + \G({\overline{\Fr}_j})-|S_k\cap S_j|, \G({\overline{\Fr}_k}) , \G({\overline{\Fr}_j})\}$, $S_k \gets S_k \cup S_j$
			
			$\rho_k \gets \frac{\left| {\bigcup\nolimits_{{\Cluster^R_l} \in {\Cluster^R \backslash \Cluster^R_k}} {{S_l \cap S_k}}} \right|}{\G({\overline{\Fr}_k})}$
			
			\label{code:merge_end}
			
			%
			
			$k \gets 	\argmax_{\Cluster^R_k \in \Cluster^R} \rho_k$, $Max \gets \rho_{k}$
			
		}
		\Return {$\Cluster^B$, $\Cluster^R$}
	\end{small}
\end{algorithm}

When the bus routes and buses are partitioned, it invokes the basic greedy method (Section~\ref{sec_4}) to find the frequency for each cluster, and merges the local frequencies for $|\Cluster^R|$ clusters as the final answer. We name this approach as \partitiongreedy. Its pseudo-code is shown in Algorithm~\ref{partitiongreedy} and its approximation ratio is analyzed in Lemma~\ref{part:appro:1}.



\begin{lemma} \label{part:appro:1}
	Given a partition $\Cluster^R$=\{$\Cluster^R_1$, $\Cluster^R_2$, $\cdots$, $\Cluster^R_i$, $\cdots$, $\Cluster^R_m$\}  of the bus route database $\Route$ and the maximum service overlap ratio $\rho$, \partitiongreedy achieves a $(1-\rho)(1-1/e)$ approximation ratio to solve the \Problem problem.
\end{lemma}

\begin{proof}
	Let $\Fr_i$ denote the solution obtained by $\textrm{\greedy}$ for cluster $\Cluster^R_i$, $\Fr^*$ denote the solution obtained by $\textrm{\partitiongreedy}$, $\OPT_i$ denote the optimal solution for cluster $\Cluster^R_i$, and $\OPT$ denote the global optimal solution. In Algorithm~\ref{partition}, it uses the lower bound of the $\G({\overline{\Fr}_k})$ to compute the upper bound of $\rho_k$ and terminates when the upper bound of $\rho_i$ for every cluster $\Cluster^R_i\in\Cluster^R$ is no greater than the given threshold $\rho$. Then we have $\rho\ge\rho_i$ for any $\Cluster^R_i\in\Cluster^R$. Recall  Section~\ref{sec_3}, the basic greedy method is proved to achieve $(1 - 1/e)$-approximation. Therefore, we have $\G(\Fr_i)\geq(1-1/e)\G(\OPT_i)$. Because of the submodularity and monotonicity of $\G$, we have $\sum_{i = 1}^{m} {{\G(\OPT_i)}}\ge\G(\OPT)$ and $\G(\Fr_i)\ge\G(\overline{\Fr}_i)$. Then, by Definition~\ref{ratio} we have:
	\begin{equation}\label{ineq:overlap1}
	\begin{aligned}
	\left| {\bigcup\nolimits_{{\Cluster^R_j} \in {\Cluster^R \backslash \Cluster^R_i}} \textsf{Serve}(\Ps, \Cluster^R_i) \cap \textsf{Serve}(\Ps, \Cluster^R_j)} \right| = \rho_i\G(\overline{\Fr}_i)\le\rho\G(\Fr_i).
	\end{aligned} 
	\end{equation}
	In addition, Inequality~(\ref{ineq:overlap2}) holds according to Definition~\ref{defn:problem}.
	\begin{equation}\label{ineq:overlap2}
	\begin{aligned}
	\left| {\bigcup\nolimits_{{\Cluster^R_j} \in {\Cluster^R \backslash \Cluster^R_i}} \textsf{Serve}(\Ps, \Cluster^R_i) \cap \textsf{Serve}(\Ps, \Cluster^R_j)} \right| \ge \G(\Fr_i) - (\G(\Fr^*)-\G(\Fr^*\backslash\Fr_i))
	\end{aligned} 
	\end{equation}
	Based on Inequality~(\ref{ineq:overlap1}) and Inequality~(\ref{ineq:overlap2}), we have $\G(\Fr^*)-\G(\Fr^*\backslash\Fr_i) \ge (1-\rho)\G(\Fr_i)$. 
	Using the principle of inclusion-exclusion, we have $ \G(\Fr^*)=\G(\Fr_1\cup\Fr_2\cup...\cup\Fr_m) \ge \sum_{i = 1}^{m} (\G(\Fr^*)-\G(\Fr^*\backslash\Fr_i)) \ge (1-\rho)\sum_{i = 1}^{m} {{\G(\Fr_i)}} \ge (1-\rho)(1-1/e)\sum_{i = 1}^{m} {{\G(\OPT_i)}} \ge (1-\rho)(1-1/e)\G(\OPT)$.
	%
	Thus, this lemma is proved. $\blacksquare$	
\end{proof}


\smallVspace
\section{Progressive Partition-based Greedy Method} \label{sec_6}
\smallVspace

%
Although \partitiongreedy improves the efficiency of basic greedy by conducting the search within each partition (though not the original route/bus database), it still suffers from a high computational cost. To be more specific, in each iteration of the greedy search (either a global search or a local search by \greedy), in order to find the one with the maximum gain, it has to recalculate the marginal gain $\G(\Fr \cup \bs)- \G(\Fr)$ for all the buses not yet scheduled. 

Motivated by this observation, we propose a progressive partition-based greedy method (\progressivepartgreedy). It selects multiple, but not only one, buses in each local greedy search iteration to cut down the total number of iterations required and hence the computation cost. The pseudo-code of \progressivepartgreedy is the same as Algorithm~\ref{partitiongreedy} except that the call of \greedy is replaced with Function~1 (\progressivegreedy) in line~\ref{part:replace} of Algorithm~\ref{partitiongreedy}. Meanwhile, we will prove that it can achieve an approximation ratio of $(1-\rho)(1-1/e-\varepsilon)$, where $\rho$ and $\varepsilon$ are tunable parameters that provide a trade-off between efficiency and accuracy. 

\setlength{\algomargin}{1.2em} 
\renewcommand{\algorithmcfname}{Function}
\setcounter{algocf}{0}
\begin{algorithm}[t]
	\caption{\progressivegreedy$(\Bs, \Route, \Number, \varepsilon)$ }
	\label{progressivegreedy}
	\begin{small}
		{\bf Input:} a bus database $\Bs$, a bus route database $\Route$, a vector $\Number$, and a parameter $\varepsilon$
		
		{\bf Output:} a bus service frequency $\Fr$
		
		Initialize $\Fr \gets \phi$, $\num \gets \sum_{i = 1}^{|\Number|} {{n_i}}$
		
		Initialize a $|\Number|$-dimension vector $\langle k_1, k_2, \cdots, k_{|\Number|}\rangle$ with zero
		
		
		Sort $\bs \in \Bs$ based on descending order of  $\G(\bs)$ 
		
		Initialize $h \gets \mathop {\max }_{\bs \in \Bs}(\G(\bs))$
		
		\While{$|\Fr|\leq\num$}
		{
			\For{each $\bs_{jl}\in\Bs$}{ \label{pro:start}
				
				\If{$|\Fr|\leq\num$}{
					
					$\G_{b_{jl}}(\Fr) \gets \G(\Fr \cup \bs_{jl})- \G(\Fr)$
					
					\If{$\G_{b_{jl}}(\Fr)\geq h$}{
						$\Fr \gets \Fr \cup \bs_{jl}$, $\Bs \gets \Bs \backslash \bs_{jl}$
						
						$k_j++$
						
						\If{$k_j\geq n_j$}
						{
							remove all bus serve the route $\rt_j$ from $\Bs$
						}
						
					}
					
					\If{$\G(b_{jl}) < h$}{ \label{pro:s1}
						\textbf{break}
					}\label{pro:e1}
				}\Else{
					\textbf{break}
				}
				
			}
			
			$h \gets  \frac{h}{1+\epsilon}$\label{pro:end}
			
		}
		\Return {$\Fr$}
	\end{small}
\end{algorithm}

As presented in Function~1, \progressivegreedy first sorts $\bs \in \Bs$ by $\G(\bs)$ and initializes the threshold $h$ to the value of $\mathop {\max }_{\bs \in \Bs}(\G(\bs))$. Then, it iteratively fetches all the buses with their marginal gains not smaller than $h$ into $\Fr$ and meanwhile lowers the threshold $h$ by a factor of $(1+\varepsilon)$ for next iteration (Lines~\ref{pro:start}-\ref{pro:end}). The iteration continues until there are $n$ buses in $\Fr$. Unlike the basic greedy method that has to check all the potential buses in $\Bs$ or a cluster of $\Bs$ in each iteration, it is not necessary for \progressivegreedy as it implements an early termination (Lines~\ref{pro:s1}-\ref{pro:e1}). Since buses are sorted by $\G(\bs)$ values, if $\G(\bs_{jl})$ of the current bus is smaller than $h$, all the buses $\bs$ pending for evaluation will have their $\G(\bs)$ values smaller than $h$ and hence could be skipped from evaluation. In the following, we first analyze the approximation ratio of Function~1 by Lemma~\ref{pro:appro:1}. Based on Lemma~\ref{pro:appro:1}, we show the approximation ratio of \progressivepartgreedy by Lemma~\ref{pro:appro:2}. 





\begin{lemma} \label{pro:appro:1}
	\progressivegreedy achieves a $(1-1/e-\varepsilon)$ approximation ratio. 
\end{lemma}

\begin{proof}
	Let $b_i$ be the bus selected at a given threshold $h$ and $\OPT$ denote the optimal local solution to the problem of selecting $n$ $buses$ that can maximize $\G$. Because of the submodularity of $\G$, we have:
	\begin{equation} \label{appro1}
	{\G_b}(\Fr) = \left\{ \begin{array}{l}
	\ge h\\
	\le h \cdot (1 + \varepsilon ) 
	\end{array} \right. \begin{array}{*{20}{l}}
	{ \textrm{if $b=b_i$}}\\
	{\textrm{if $b\in \OPT\backslash(\Fr\cup b_i)$}},
	\end{array}
	\end{equation}
	where $\Fr$ is the current partial solution. Equation~(\ref{appro1}) implies that $\G_{\bs_i}(\Fr)\geq\G_\bs(\Fr)/(1+\varepsilon)$ for any $\bs \in \OPT\backslash \Fr$. Thus, we have $\G_{\bs_i}(\Fr)\geq \frac{1}{(1+\varepsilon)|\OPT\backslash \Fr|}\sum\nolimits_{\bs \in \OPT\backslash \Fr}\G_\bs(\Fr)   \geq \frac{1}{(1+\varepsilon)n}\sum\nolimits_{\bs \in \OPT\backslash \Fr}\G_\bs(\Fr)$.
	Let $\Fr_i$ denote the partial solution that $\bs_i$ has been included and $\bs_{i+1}$ be the bus selected at the $(i+1)$th step. Then we have $\G(\Fr_{i+1})-\G(\Fr_{i}) =\G_{b_{i+1}}(\Fr_{i}) \geq \frac{1}{(1+\varepsilon)n}\sum\nolimits_{\bs \in \OPT\backslash \Fr_i}\G_\bs(\Fr_i)
	\geq \frac{1}{(1+\varepsilon)n}(\G(\OPT\cup\Fr_{i})-\G(\Fr_{i})) \geq \frac{1}{(1+\varepsilon)n}(\G(\OPT)-\G(\Fr_{i}))$. 
	
	The solution $\Fr^*$ obtained by Function~\ref{progressivegreedy} with $|\Fr^*|= n$. Using the geometric series formula, we have $\G(\Fr^{*})\ge\left( {1 - \left( 1- \frac{{1}}{{(1 + \varepsilon )n}}\right)^n} \right)\G\left( \OPT \right)  \ge\left( {1 -  {e^{\frac{{ - n}}{{(1 + \varepsilon )n}}}}} \right)\G\left( \OPT \right) = \left( {1 - {e^{\frac{{ - 1}}{{(1 + \varepsilon )}}}}} \right)\G\left( \OPT \right) \ge\left( {(1-1/e-\varepsilon)} \right)\G\left( \OPT \right)$. 
	%
	Hence, the lemma is proved. $\blacksquare$
\end{proof}

\begin{lemma} \label{pro:appro:2}
	Given a partition $\Cluster^R$=\{$\Cluster^R_1$, $\Cluster^R_2$, $\cdots$, $\Cluster^R_i$, $\cdots$, $\Cluster^R_m$\}  of the bus route database $\Route$ and the maximum service overlap ratio $\rho$, \progressivepartgreedy achieves a $(1-\rho)(1-1/e-\varepsilon)$ approximation ratio to solve the \Problem problem.
\end{lemma}

\begin{proof}
	Based on Lemma~\ref{pro:appro:1}, this proof is similar to the proof of Lemma~\ref{part:appro:1}, so we omit it due to space limit. $\blacksquare$
\end{proof}


\vspace{-.8cm}
\begin{table}[]
	\caption{Statistics of datasets}\vspace{-.8em}
	\smallVspace
	\label{datasets}
	\begin{center}
		\begin{tabular}{|p{2cm}<\centering|p{2cm}<\centering|p{2cm}<\centering|p{2.5cm}<\centering|}
			\hline
			Database & Amount & AvgDistance & AvgTravelTime \\ \hline
			$\Bs$ & 451k & N.A. & N.A. \\ \hline
			$\Route$ & 396 & 19.91km & 5159s \\ \hline
			$\Ps$ & 28m & 4.2km & 1342s \\ \hline
		\end{tabular}
	\end{center}
\end{table}
\vspace{-1.2cm}

\smallVspace
\section{Experiment} \label{sec_7}
\smallVspace

In this section, we first explain the experimental setup; we then conduct sensitivity tests to tune the parameters to their reasonable settings, as our algorithms have several tunable parameters; we finally report the performance, in terms of effectiveness, efficiency, and scalability, of all the algorithms. 

\noindent
\textbf{Datasets.} We crawl the real bus routes $(\Route)$ from transitlink\footnote{\url{https://www.transitlink.com.sg/eservice/eguide/service_idx.php}} in Singapore. Each route is represented by the sequence of bus stop IDs it passes sequentially, together with the distance between two consecutive bus stops. The travel time from a stop to another stop via a route $\rt_i$ is estimated by the ratio of the distance between those two stops along the route to the average bus speed of the route. We use bus touch-on record data (shown later) to find the average travel speed of a particular bus line.
%
%
For the passenger database $(\Ps)$, due to the exhibit regular travel patterns of passengers~\cite{tian2018using}, we use the real bus touch-on record data in a week of April 2016 in Singapore, which is obtained from the authors of~\cite{tian2018using} and contains 28 million trip records. Each trip record includes the IDs/timestamps of the boarding and alighting bus stops, the bus route, and the trip distance. We assume passengers spend $x$ minutes waiting for their buses, with $x$ following a random distribution between 1 and 5 minutes. 
%
Then, we generate the bus candidate set $(\Bs)$ based on the route and service time range. For each route, we use buses that depart every minute between 5am and 12am as the superset of candidate buses. The statistics of those datasets are shown in Table~\ref{datasets}.

\noindent
\textbf{Parameters.}
Table~\ref{parameter} lists the parameter settings, with values in bold being default. In all experiments, we vary one parameter and set the rest to their defaults. We assume all bus routes require the same number of bus departures in our study. 
%
Notation $\langle 20 \rangle$ represents the vector $\langle 20$, $\cdots$, $20 \rangle$ for brevity. 
\begin{table}[t]
	\caption{Parameter settings}
	\bigVspace
	\label{parameter}
	\begin{center}
		\begin{tabular}{|p{6.4cm}<\centering|p{5cm}<\centering|}
			\hline
			Parameter & Values \\ \hline
			number of bus departures $\Number =\langle n_1, n_2, \cdots\rangle$ & $\langle 10 \rangle$, $\langle 20 \rangle$, $\bm{\langle 30 \rangle}$, $\langle 40 \rangle$, $\langle 50 \rangle$   \\ \hline
			total passenger number $|\Ps|$ & 100k, 200k, \textbf{300k}, 400k, 500k \\ \hline
			waiting time threshold $\theta$ & 1min, 2min, \textbf{3min}, 4min, 5min \\ \hline
			tunable parameter used by\progressivepartgreedy $\varepsilon$ & $10^{-4}$,  $10^{-3}$, $\bm{10^{-2}}$, $10^{-1}$\\ \hline
			controlling threshold used by \partitiongreedy $\rho$ & 0.1, \textbf{0.2}, 0.3, 0.4 \\ \hline
		\end{tabular}
	\end{center}
	\bigVspace
	\smallVspace
\end{table}

\noindent
\textbf{Algorithm.} To the best of our knowledge, this is the first work to study the \Problem problem, and thus no previous work is available for direct comparison. In particular, we compare the following five methods. \underline{\fixint} that fixes the time interval between two bus departures as $\lfloor$(service time range) / (bus number)$\rfloor$ for each line and chooses the bus that departures at 5am as the first bus; \underline{\topk} that picks top-$k$ buses, which could serve the most number of passengers ($k = n_i$); 
%
\underline{\greedy}, \underline{\partitiongreedy}, and \underline{\progressivepartgreedy}, i.e., Algorithm~\ref{greedy}, Algorithm~\ref{partitiongreedy}, and the progressive partition-based method proposed in this paper.

\noindent
\textbf{Performance measurement.}
We adopt the \emph{total running time} of each algorithm and the \emph{total served passenger number (SPN)} of the scheduled buses as the main performance metrics. We randomly choose 5 million passengers from a week of data and pre-process the passenger dataset to build the index, which takes $5,690$ seconds and occupies $585$MB disk space. 
%
%
Each experiment is repeated ten times, and the average result is reported.

\noindent
\textbf{Setup.}
All codes are implemented in C++. Experiments are conducted on a server with 24 Intel X5690 CPU and 140GB memory running CentOS release 6.10. We will release the code publicly once the paper is published.


\noindent
\textbf{Parameter Sensitivity Test - $\theta$.} The impact of waiting time threshold $\theta$ on the running time and SPN are reported in Figure~\ref{fig:efficiency:theta} and Figure~\ref{fig:effectiveness:theta}, respectively. Parameter $\theta$ has an almost-zero impact on the running time. On the other hand, it affects SPN. As $\theta$ increases, all the algorithms are able to serve more passengers, which is consistent with our expectations. We set $\theta=3$, the mean value.
%

\noindent
\textbf{Parameter Sensitivity Test - $\rho$.} The impact of parameter $\rho$ on the running time and SPN are reported in Figure~\ref{fig:efficiency:rho} and Figure~\ref{fig:effectiveness:rho}, respectively. It has a positive impact on the running time performance but a negative impact on SPN. As $\rho$ increases its value, \partitiongreedy and \progressivepartgreedy both incur shorter running time but serve less number of passengers. We choose  $\rho=0.2$ as the default setting.
%

\noindent
\textbf{Parameter Sensitivity Test - $\varepsilon$.} Parameter $\varepsilon$ only affects \progressivepartgreedy. It controls the trade-off between efficiency and accuracy. As $\varepsilon$ increases its value, \progressivepartgreedy incurs shorter running time and serves less number of passengers, as reported in  Figure~\ref{fig:efficiency:rho} and Figure~\ref{fig:effectiveness:epsilon}, respectively. We choose  $\varepsilon=0.01$ as the default setting.

\begin{figure}[t]
	\centering
	\includegraphics[width=1\textwidth]{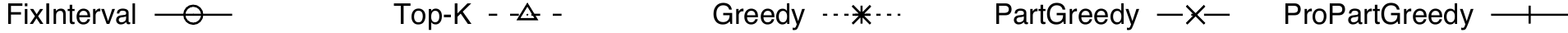}
	\vspace{-0.3cm}
	\hspace{-10pt}
	\subfigure[Running time vs. $\theta$]
	{\label{fig:efficiency:theta}\includegraphics[width=0.345\textwidth]{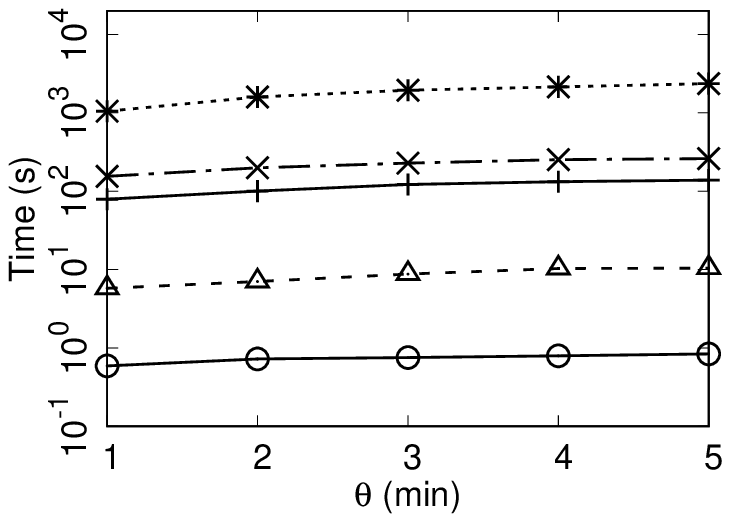}}
	\hspace{-10pt}
	\subfigure[Running time vs. $\rho$]
	{\label{fig:efficiency:rho}\includegraphics[width=0.345\textwidth]{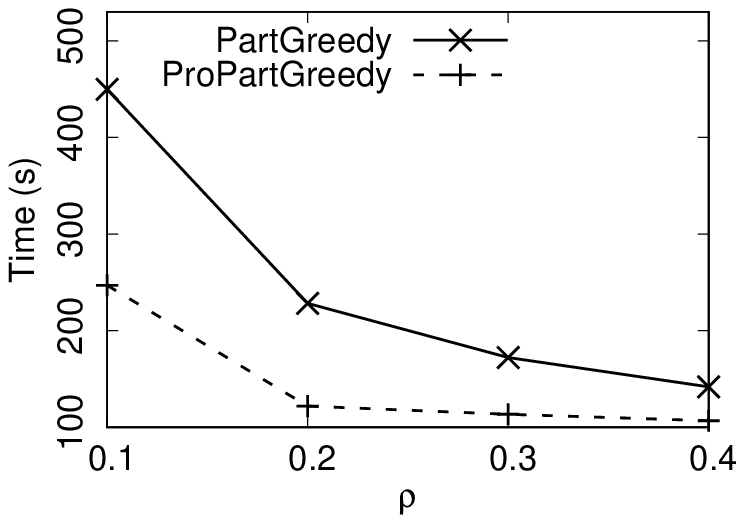}}
	\hspace{-10pt}
	\subfigure[Running time vs.$\varepsilon$]
	{\label{fig:efficiency:epsilon}\includegraphics[width=0.345\textwidth]{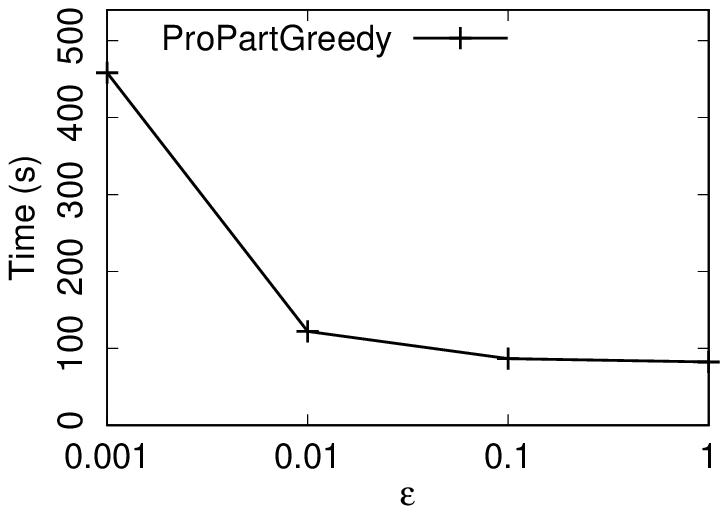}}\\
	\hspace{-9pt}
	\subfigure[SPN vs. $\theta$]
	{\label{fig:effectiveness:theta}\includegraphics[width=0.345\textwidth]{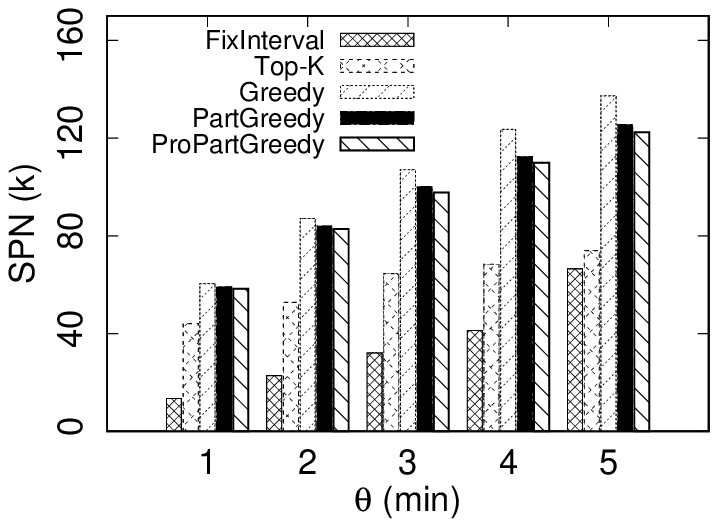}}
	\hspace{-10pt}
	\subfigure[SPN vs. $\rho$]
	{\label{fig:effectiveness:rho}\includegraphics[width=0.345\textwidth]{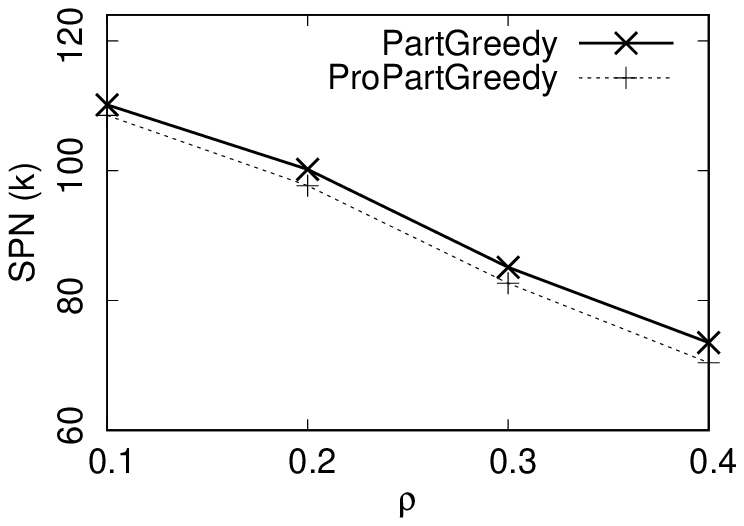}}
	\hspace{-10pt}
	\subfigure[SPN vs. $\varepsilon$]
	{\label{fig:effectiveness:epsilon}\includegraphics[width=0.345\textwidth]{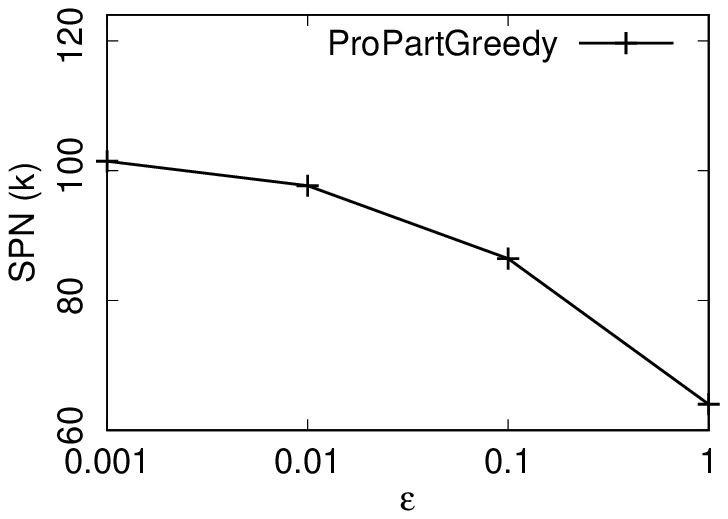}}
	\smallVspace
	\caption{Effect of parameters}
	\label{fig:para}
	\bigVspace
\end{figure}
\begin{figure}[t]
	\centering
	\hspace{-10pt}
	\subfigure[SPN vs. $\Number$]
	{\label{fig:effect:bus}\includegraphics[width=0.345\textwidth]{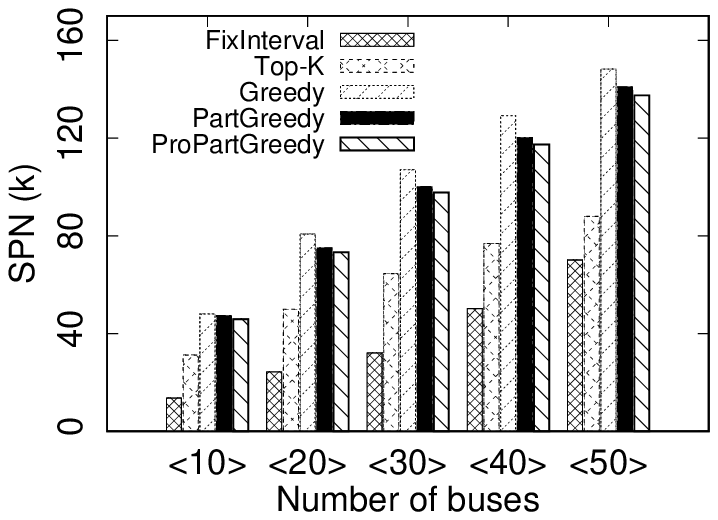}}
	\hspace{10pt}
	\subfigure[SPN vs. $|\Ps|$]
	{\label{fig:effect:passenger}\includegraphics[width=0.345\textwidth]{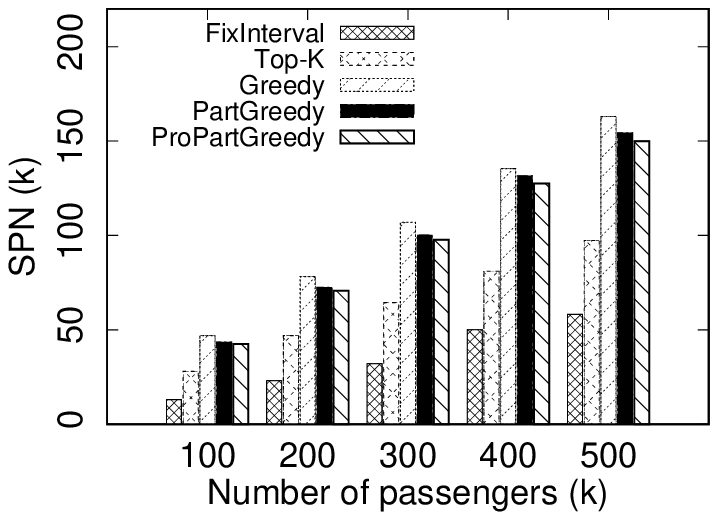}}
	\smallVspace
	\caption{Effectiveness Study: SPN vs. $\Number$ or $|\Ps|$}
	\label{fig:effectiveness}
	\bigVspace
\end{figure}

\noindent
\textbf{Effectiveness Study.} We report the effectiveness of different algorithms in Figure~\ref{fig:effectiveness}. We observe that 
%
(1) \fixint is most \emph{ineffective}; (2) the three algorithms proposed in this work perform much better than the other two, e.g., \progressivepartgreedy doubles (or even triples in some cases) the SPN of \fixint;  and (3) \greedy performs the best while \partitiongreedy and \progressivepartgreedy achieve comparable performance (only up to 9.4\% below that of \greedy).
%


\noindent
\textbf{Efficiency Study.} 
Figure \ref{fig:efficiency} shows the running time of each method w.r.t. varying $\Number$ and $|\Ps|$. We have two main observations. (1) The time gap among \greedy, \partitiongreedy and \progressivepartgreedy becomes more significant with the increase of $\Number$. This could be the increase of $\Number$ causes an increase in the number of clusters and $n_{min}$. On the other hand, \partitiongreedy and \progressivepartgreedy only need to scan one cluster when selecting buses. (2) The improvement of \partitiongreedy and \progressivepartgreedy over \greedy decreases with the increase of $|\Ps|$. This is because the overlap between clusters increases with the increase of $|\Ps|$, which leads to a reduction in the number of clusters and an increase in partition time.

\noindent
\textbf{Scalability Study.}
To evaluate the scalability of our methods, we vary $\Number$ from $\langle 100 \rangle$ to $\langle 500 \rangle$, and $|\Ps|$ from 1 million to 5 million. From Figure~\ref{fig:scalability:bus}, we find that the efficiency of \greedy is more sensitive to $\Number$, as compared to \partitiongreedy and \progressivepartgreedy. It's worth noting that the results are omitted for \greedy when it cannot terminate within $10^4$ seconds. As shown in Figure~\ref{fig:scalability:passenger}, \partitiongreedy and \progressivepartgreedy are about ten times faster than \greedy when $|\Ps|$ is varying. 

\begin{figure}[t]
	\centering
	\vspace{-0.3cm}
	\includegraphics[width=1\textwidth]{title5.pdf}
	\vspace{-0.4cm}
	\hspace{-10pt}
	\subfigure[Running time vs. $\Number$]
	{\label{fig:efficiency:bus}\includegraphics[width=0.345\textwidth]{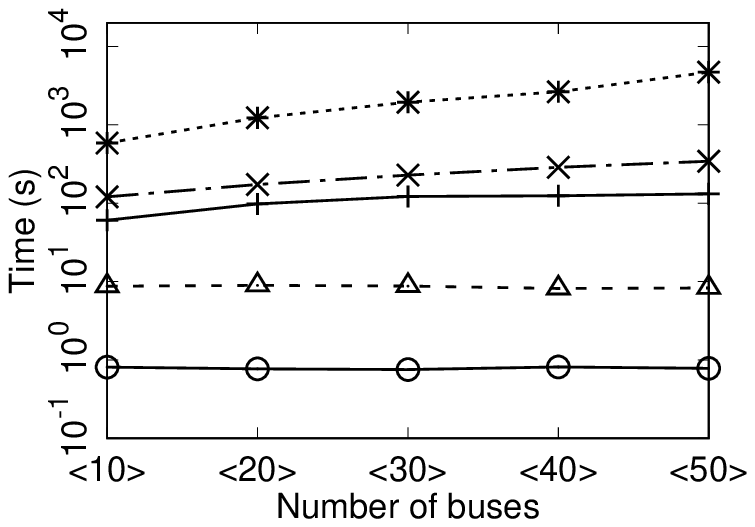}}
	\hspace{10pt}
	\subfigure[Running time vs. $|\Ps|$]
	{\label{fig:efficiency:passenger}\includegraphics[width=0.345\textwidth]{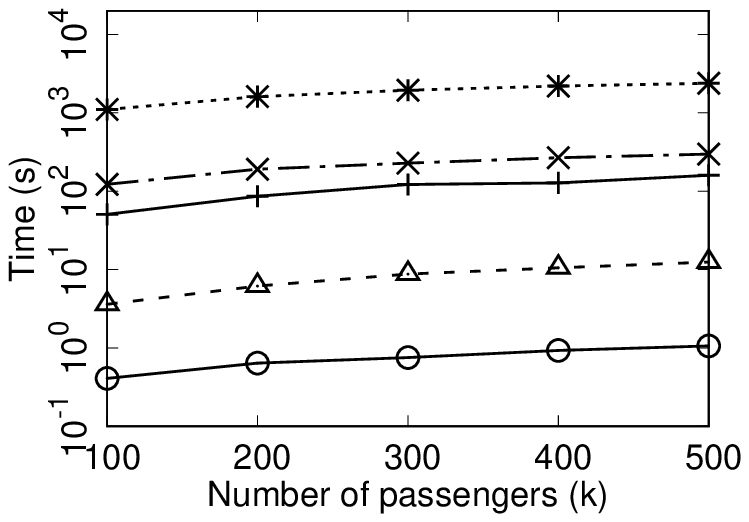}}
	\caption{Efficiency Study: Total Running Time vs. $\Number$ or $|\Ps|$}
	\label{fig:efficiency}
	\smallVspace
\end{figure}
\begin{figure}[t]
	\centering
	\hspace{-10pt}
	\subfigure[Running time vs. $\Number$]
	{\label{fig:scalability:bus}\includegraphics[width=0.345\textwidth]{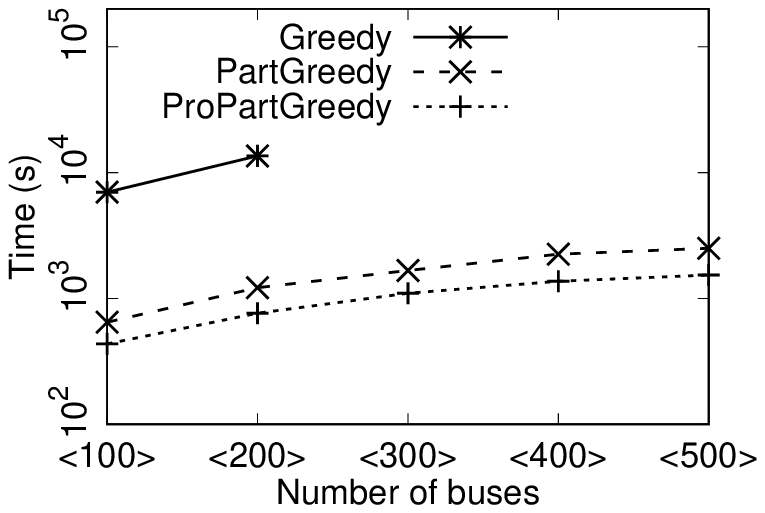}}
	\hspace{10pt}
	\subfigure[Running time vs. $|\Ps|$]
	{\label{fig:scalability:passenger}\includegraphics[width=0.345\textwidth]{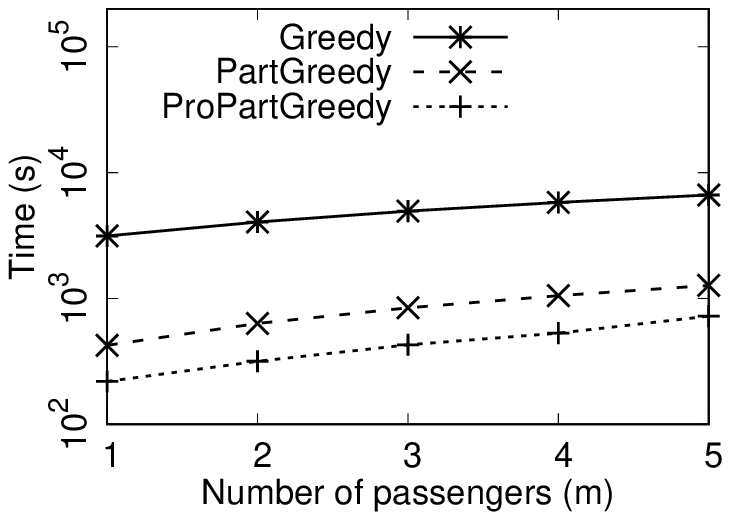}}
	\smallVspace
	\caption{Scalability Study}
	\label{fig:scal}
	\bigVspace
\end{figure}

\smallVspace
\section{Conclusion} \label{sec_8}
\smallVspace
In this paper we studied the bus frequency optimization problem considering user satisfaction for the first time. Our target is to schedule the buses in such a way that the total number of passengers who could receive their bus services within the waiting time threshold is maximized. We showed that this problem is NP-hard, and proposed three approximation algorithms with non-trivial theoretical guarantees. Lastly, we conducted experiments on real-world datasets to verify the efficiency, effectiveness, and scalability of our methods.

\noindent
\textbf{Acknowledgements.} Zhiyong Peng is supported in part by the National Key Research and Development Program of China (Project Number: 2018YFB1003400), Key Project of the National Natural Science Foundation of China (Project Number: U1811263) and the Research Fund from Alibaba Group. Zhifeng Bao is supported in part by ARC DP200102611, DP180102050, NSFC 91646204, and a Google Faculty Award. Baihua Zheng is supported in part by Prime Minister’s Office, Singapore under its International Research Centres in Singapore Funding Initiative.

{\small
\bibliographystyle{splncs04}
\bibliography{ref}
}
\end{document}